\documentclass[preprint,12pt]{elsarticle}
\usepackage{amsmath}
\usepackage{amssymb}
\usepackage{multirow}
\usepackage{enumerate}
\usepackage{cases}
 \usepackage{amsthm}

\usepackage{mathrsfs}

\makeatletter

\newcommand {\Rmnum} [1] {\expandafter \@slowromancap \romannumeral #1@}
\makeatother
\newtheorem{Definition}{Definition}

\newtheorem{Remark}{Remark}

\newtheorem{Proposition}{Proposition}

\biboptions{compress}

\begin{document}

\begin{frontmatter}

\title{Bit-oriented quantum public-key encryption}


\author{Chenmiao Wu $^{1,2,3}$}
\author{Li Yang$^{1,2}$\corref{1}}
\cortext[1]{Corresponding author email: yangli@iie.ac.cn}
\address{1.State Key Laboratory of Information Security, Institute of Information Engineering, Chinese Academy of Sciences, Beijing 100093, China\\
2.Data Assurance and Communication Security Research Center,Chinese Academy of Sciences, Beijing 100093, China\\
3.University of Chinese Academy of Sciences, Beijing, 100049, China}

\begin{abstract}
We propose a bit-oriented quantum public-key scheme which uses Boolean function as private-key and randomly changed pairs of quantum state and classical string as public-keys. Contrast to the typical classical public-key scheme, one private-key in our scheme corresponds to an exponential number of public-keys. The goal of our scheme is to achieve information-theoretic security, and the security analysis is also given.
\end{abstract}


\end{frontmatter}



\section{Introduction}
  Public-key encryption (PKE) is one of the important branches of cryptography, and has been widely applied into various fields. Classical public-key encryption is based on one-way function with the hardness of solving the computational difficult problem. Since the investigation to quantum Turing machine was started and Shor's algorithm as well as Grover's algorithm were put forward, quantum public-key encryption (QPKE) emerged as times require.

  According to whether the encrypted messages and keys are in quantum states or not, we classify QPKE into four types. The first type is that messages and keys are both classical. The Knapsack-based scheme\cite{Ok00} proposed by Okamato uses classical key to encrypt and to decrypt classical messages, but the participants are all quantum probabilitic polynomial Turing machine. The second type encrpts quantum messages with classical keys. The McEliece QPKE scheme proposed in \cite{YL03} was belong to this type. This scheme was based on a classical NP-complete problem related with finding a code word of a given weight in a linear binary code. With this foudation, Yang et al. \cite{YL11} gave the definition of induced trapdoor one-way quantum function. Fujita\cite{Fu12} also constructed McEliece QPKE relied on the difficulty of NPC problem. The third type refers to classical messages and quantum keys. Gottesman and Chuang \cite{GC01} were the first to construct quantum states as public-keys. Kawachi\cite{Ka05,Ka12} investigated the cryptographic property "computational indistinguishability" of two quantum states generated via fully flipped permutations and proposed a QPKE based on it. Nikolopoulos \cite{GMN08} constructed a QPKE from the perspective of single-qubit rotation and trapdoor one-way function. \cite{YYX13} also proposed a information-theoretic secure QPKE designed with conjugate coding single-photon string. The fourth is with quantum messages and quantum keys. Gottesman \cite{Gottesman05} proposed a QPKE based on teleportation with information-theoretic security. Liang et al. \cite{LY12} combined the basic idea of \cite{PY12} and quantum perfect encryption to construct an information-theoretically secure QPKE. Kawachi and Pormann \cite{KP08} presented another kind of quantum message-oriented public-key encryption protocol with quantum public-key, but they show this scheme is  bounded information-theoretic secure.

 In this paper, we invetigate the way to construct bit-oriented public-key encryption scheme. Firstly, we overview some definition about information-theoretically security about QPKE and quantum perfect encryption. Second, we present a new bit-oriental public-key encryption scheme. Finally, we investigate the scheme's security against attack to the key and attack to the encryption and prove that the scheme is information-theoretically secure.

\section{Preliminaries}
\subsection{Information-theoretic security}
Goldreich \cite{Go04} defined the ciphertext indistinguishability in classical PKE: for every polynomial-size circuit family ${c_{n}}$, every positive polynomial $p(\cdot)$, all sufficiently large n, and every $x, y\in \{0,1\}^{*}$, satisfies:
\begin{eqnarray}
\Big|\Pr\big[C_{n}(G(1^{n}),E_{G(1^{n})}(x))=1\big]-\Pr\big[C_{n}(G(1^{n}),E_{G(1^{n})}(y))=1\big]\Big|<\frac{1}{p(n)}.
\end{eqnarray}
Similar to the ciphertext indistinguishability in classical context, \cite{YXL10} extend the concept to QPKE and presented the definition of information-theoretic quantum ciphertext-indistinguishability for public-key encryption of classical messages:
\begin{Definition}
A quantum public-key encryption scheme of classical messages is information-theoretically ciphertext-indistinguishable, if for every quantum circuit family $\{C_{n}\}$, every positive polynomial $p(\cdot)$, all sufficiently large n, and any two bit-string $x, y\in \{0,1\}^{*}$, satisfies:
\begin{eqnarray}
\Big|\Pr\big[C_{n}(G(1^{n}),E_{G(1^{n})}(x))=1\big]-\Pr\big[C_{n}(G(1^{n}),E_{G(1^{n})}(y))=1\big]\Big|<\frac{1}{p(n)}
\end{eqnarray}
where $G$ is key generation algorithm, E is a quantum encryption algorithm, and the ciphertext $E(x)$, $E(y)$ are quantum states.
\end{Definition}
\cite{YXL10} also proved that a quantum public-key encryption scheme is information-theoretically secure if the trace distance between any two cihpertexts is $O(\frac{1}{2^{n}})$, and Eq. (1) holds.
\begin{Definition}
For all plaintext $x$ and $y$, let the density operators of ciphertext $E(x)$ and $E(y)$
 be $\rho_{x}$ and $\rho_{y}$ respectively. A quantum public-key encryption scheme is said to be information-theoretically secure, if for every positive polynomial $p(\cdot)$ and every sufficiently large n, satisfies:
\begin{eqnarray}
D(\rho_{x}, \rho_{y})<\frac{1}{p(n)}
\end{eqnarray}
\end{Definition}

\subsection{Quantum perfect encryption}
Assume a set of oeprations $U_{k}$, each $U_{k}$ is a $2^{n}\times 2^{n}$ unitary matrix. The ciphertext state of n-qubit quantum message $\rho$ is $\rho_{C}$. $k$ refers to private-key, each $k$ is chosen with probability $p_{k}$. To encrypt, applying $U_{k}$ to $\rho$:
\begin{eqnarray}
\rho_{c}&=&U_{k}\rho U^{\dagger}_{k}
\end{eqnarray}
The private-key owner can use $U^{\dagger}_{k}$ to decrypt:
\begin{eqnarray}
\rho&=&U^{\dagger}_{k}\rho U_{k}
\end{eqnarray}
As defined in \cite{BR03}: for every input state $\rho$, the output state is an ultimately mixed state:
\begin{eqnarray}
\sum\limits_{k}p_{k}U_{k}\rho U^{+}_{k}=\frac{I}{2}.
\end{eqnarray}
That is:

$\{p_{k}=\frac{1}{2^{2n}}, U_{k}=U^{\alpha}_{1}U^{\beta}_{2}, k=(\alpha,\beta), \alpha,\beta\in\{0,1\}^{n}\}$ is a quantum perfect encryption, where $U_{k}=U^{\alpha}_{1}U^{\beta}_{2}$ constitute a complete orthogonal basis.


\subsection{Simple description of QPKE in \cite{PY12}}
The QPKE in \cite{PY12} consists of three phases: key generation, encryption and decryption.

Suppose that $\Omega_{n}=\{k\in Z_{2^{n}}|W_{H}(k)~ is~ odd\}$, $\Pi_{n}=\{k\in Z_{2^{n}}|W_{H}(k)~ is ~even\}$. $W_{H}$ denotes the hamming weight of $k$. In this scheme defined two n-qubit states:

\begin{eqnarray}
\rho^{0}_{k,i}=\frac{1}{2}(|i\rangle+|i\oplus k\rangle)(\langle i|+\langle i\oplus k|)
\end{eqnarray}
and
\begin{eqnarray}
\rho^{1}_{k,i}=\frac{1}{2}(|i\rangle-|i\oplus k\rangle)(\langle i|-\langle i\oplus k|),
\end{eqnarray}
where $i\in Z_{2^{n}}$, $k\in \Omega_{n}$.
\begin{bfseries}
\flushleft{[Key Generation]}
\end{bfseries}
\begin{enumerate}
\item Bob randomly selects a Boolean Function as his private-key from the mapping set $F:\{0,1\}^{m}\mapsto \{0,1\}^{n}$, and also chooses a bit string $s$, $s\in Z_{2^{m}}$;
\item Bob computes $F(s)=k$, and then prepares quantum state $\rho^{0}_{k,i}$ according to the string $k\in\Omega_{n}$, $i\in Z_{2^{n}}$;
\item Bob sends his public-key $(s,\rho^{0}_{k,i})$ to the public register.
\end{enumerate}
\begin{bfseries}
\flushleft{[encryption]}
\end{bfseries}
\begin{enumerate}
\item  Ailce download Bob's public-key from the public register;
\item if she wants to send message "0", she sends $(s,\rho^{0}_{k,i})$ to Bob; if the message is "1", she does $Z^{\otimes n}$ on $\rho^{0}_{k,i}$ to acquire $\rho^{1}_{k,i}$, then sends $(s,\rho^{1}_{k,i})$ to Bob.
\end{enumerate}
\begin{bfseries}
\flushleft{[decryption]}
\end{bfseries}
\begin{enumerate}
\item Receiving the message send from Alice, Bob uses his private-key $F$ to compute $F(S)=k$;
\item Using a bit "1" in $k$ as controlled bit to do CONT operation to other bits in the quantum state, then measuring the quantum state with basis $\{|+\rangle, |-\rangle\}$ to get the message.
\end{enumerate}
\subsection{Attack method}
\cite{YYX13} presented an effective attack for the scheme descrpted in 2.4. The concrete steps are as follows:
\begin{enumerate}
\item when an attacker gets user's public-key $(s,\rho^{0}_{k,i})$, he does random Hadamard transform to the quantum state $\rho^{0}_{k,i}$:
\begin{eqnarray}
H^{\otimes n}[\frac{1}{\sqrt{2}}(|i\rangle+|i\oplus k\rangle)]&=&\frac{1}{\sqrt{2^{n+1}}}\sum\limits_{y}((-1)^{y\cdot i}|y\rangle+(-1)^{y\cdot (i\oplus k)}|y\rangle)\nonumber\\
&=&\frac{1}{\sqrt{2^{n+1}}}\sum\limits_{y}(-1)^{y\cdot i}(1+(-1)^{y\cdot k})|y\rangle
\end{eqnarray}
\item Attacker measures $|y\rangle$, and gets $y_{0}$.
\end{enumerate}

$y_{0}$ satisfies $y_{0}\cdot k=0$. That is to say, it satisfies the equation $y_{0}\cdot F(s)=0$. The Boolean function is n-input and n-output, it can be expressed as:
\begin{eqnarray}
F(s)=(F^{(1)}(s),\ldots,F^{(n)}(s))
\end{eqnarray}
The minor term expression of every $F^{(i)}(s)$ is:
\begin{eqnarray}
F^{(j)}(s)&=&(s^{a_{j11}}_{1}\cdot \ldots \cdot s^{a_{j1n}}_{n})\oplus\ldots\oplus(s^{a_{jp(n)1}}_{1}\cdot \ldots \cdot s^{a_{jp(n)n}}_{n})
\end{eqnarray}
Since that $x^{a}=xa\oplus a\oplus1$, we rewrite the Boolean function in linear expression:
\begin{eqnarray}
F^{(j)}(s)&=&\bigoplus^{p(n)}\limits_{\alpha=1}(\prod^{n}\limits_{\beta=1}(s_{\beta}a_{j\alpha\beta}+a_{j\alpha\beta}+1))
\end{eqnarray}
So the attacker gets the equation about $F(s)$, he will be able to acquire information about user's private-key $F$.

\section{bit-oriented public-key encryption}
\subsection{preparition of quantum state}
Firstly, we take into account how to construct the quantum state used in this scheme.
The quantum state is $\rho^{0}_{k}$:
 \begin{eqnarray}
\rho^{0}_{k}=\frac{1}{2}(|0\rangle+|k\rangle)(\langle0|+\langle k|)
\end{eqnarray}                  
The concrete steps of preparing quantum state $\rho^{0}_{k}$ is as follows:
 \begin{enumerate}
\item Prepare quantum state $|0\rangle\otimes|0\rangle$ and two quantum registers. Applying Hadamard transform to the first register, the state of the whole system becomes:
\begin{eqnarray}
(H\otimes I^{\otimes n})(|0\rangle\otimes|0\rangle)=\frac{|0\rangle+|1\rangle}{\sqrt{2}}\otimes |0\rangle
\end{eqnarray}
\item Define controlled-k operator $C_{k}$ as: $C_{k}|0\rangle|0\rangle=|0\rangle|0\rangle$, $C_{k}|1\rangle|0\rangle=|1\rangle|k\rangle$, $C_{k}$ can be realized via a group of CNOT operations.
We get the state:
\begin{eqnarray}
\frac{1}{\sqrt{2}}(|0\rangle|0\rangle+|1\rangle|k\rangle)
\end{eqnarray}
\item Use one of the non-zero bits in $k$ to do CNOT to the first register, then obtain:
\begin{eqnarray}
\frac{1}{\sqrt{2}}(|0\rangle|0\rangle+|0\rangle|k\rangle)
\end{eqnarray}
\end{enumerate}
Finally, we get the state $\rho^{0}_{k,i}$.

$\rho^{1}_{k,i}$ can be obtained by applying $Z^{\otimes n}$ on $\rho^{0}_{k,i}$.
\begin{Remark}
 \rm{ There is another way to prepare $\rho^{0}_{k,i}$:}
\begin{enumerate}
\item  \rm{ prepare n-qubit quantum state $|0\rangle$, and select the jth qubit $"0_{j}"$ which at the same time satisfies $k_{j}=1$  in $|0\rangle$ to do Hadamard transform:}
\begin{eqnarray}
|0_{1},\ldots,0_{j-1}\rangle(\frac{|0_{j}\rangle+|1_{j}\rangle}{\sqrt{2}})|0_{j+1},\ldots,0_{n}\rangle\nonumber,
\end{eqnarray}
\item  \rm{ use all the "1" in $k$ to do CNOT on the above quantum state, we acquire:}
\begin{eqnarray}
\frac{1}{\sqrt{2}}(|0\rangle+|k\rangle)\nonumber. \end{eqnarray} \end{enumerate}
 \rm{ Compared with the previous one, this method is more efficient because of that the total number of CNOT operation decreased by two.}
\end{Remark}

\subsection{detail description of bit-oriented public-key encryption}
Let $k_{1}\in\Omega_{n}$, $k_{1}=\{k_{11}, \ldots, k_{1n}\}$, $k_{2}=\{k_{21}, \ldots, k_{2n}\}$ and $k_{3}=\{k_{31}, \ldots, k_{3n}\}$ are bit strings, where each element is in $\{0, 1\}$. We denote $H^{\otimes k}=H^{k_{1}}\otimes\cdots\otimes H^{k_{n}}$, where $H^{0}=I$, $H^{1}=H$. The definition is similar with $Y^{\otimes k}$
\begin{bfseries}
\flushleft{[Key Generation]}
\end{bfseries}
\begin{enumerate}

\item Bob selects randomly two Boolean function $F=(F_{1}, F_{2}, F_{3})$ from $F:\{0,1\}^{m}\mapsto \{0,1\}^{n}$ as his private-key, and also chooses two string $s=(s_{1}, s_{2}, s_{3})$ randomly from the set $\{0,1\}^{m}$ to do the computation $F_{1}(s_{1})=k_{1}$, $F_{2}(s_{2})=k_{2}$, $F_{3}(s_{3})=k_{3}$ to get $k=(k_{1}, k_{2}, k_{3})$. If $k_{1}\notin \Omega_{n}$, Bob selects $s_{1}$ again until $k_{1}\in \Omega_{n}$;
\item Bob uses $k_{1}$ to prepare quantum state $\rho^{0}_{k_{1}}$,and then does Hadamard transform on $\rho^{0}_{k_{1}}$ to get the quantum part of public-key according to $k_{2}$:
\begin{eqnarray}
\rho^{0}_{F(s)}=\frac{1}{2}Y^{\otimes k_{3}}H^{\otimes k_{2}}(|0\rangle+|k_{1}\rangle)(\langle 0|+\langle k_{1}|)H^{\otimes k_{2}}Y^{\otimes k_{3}}
\end{eqnarray}
\item Bob stores $(s, \rho^{0}_{F(s)})$ in the quantum registers as his public-key.
\end{enumerate}
\begin{bfseries}
\flushleft{[encryption]}
\end{bfseries}
\begin{enumerate}
\item Alice downloads Bob's public-key from quantum registers;
\item if Alice wants to send message 0, she sends $(s,\rho^{0}_{F(s)})$ to Bob; if she sends 1, she does transform as bellow:
\begin{eqnarray}
\rho^{1}_{F(s)}&=&Y^{\otimes n}\rho^{0}_{F(s)}(Y^{\otimes n})^{\dagger}\nonumber\\
&=&\frac{1}{2}Y^{\otimes n}Y^{\otimes k_{3}}H^{\otimes k_{2}}(|0\rangle+|k_{1}\rangle)(\langle 0|+\langle k_{1}|)H^{\otimes k_{2}}Y^{\otimes k_{3}}Y^{\otimes n}\nonumber\\
&=&\frac{1}{2}Y^{\otimes k_{3}}H^{\otimes k_{2}}(|1\rangle+(-1)^{p(k_{1})}|\overline{k_{1}}\rangle)(\langle 1|+(-1)^{p(k_{1})}\langle \overline{k_{1}}|)H^{\otimes k_{2}}Y^{\otimes k_{3}}\nonumber\\
&=&\frac{1}{2}Y^{\otimes k_{3}}H^{\otimes k_{2}}(|1\rangle-|\overline{k_{1}}\rangle)(\langle 1|-\langle \overline{k_{1}}|)Y^{\otimes k_{3}}H^{\otimes k_{2}}, \end{eqnarray}
 where the length of $n$ is even,
\item Alice sends $(s, \rho^{1}_{F(s)})$ to Bob;
\end{enumerate}
\begin{bfseries}
\flushleft{[Decryption]}
\end{bfseries}
\begin{enumerate}
\item Bob uses $s=(s_{1}, s_{2}, s_{3})$ and $F=(F_{1}, F_{2}, F_{3})$ to compute $k=(k_{1}, k_{2}, k_{3})$;
\item Bob uses $k_{2}, k_{3}$  to remove transforms on the ciphertext state $\rho^{b}_{F(s)}$:
 \begin{eqnarray}
 \rho^{b}_{k}=(Y^{\otimes k_{3}}H^{\otimes k_{2}})^{\dagger}\rho^{b}_{F(s)}Y^{\otimes k_{3}}H^{\otimes k_{2}}
 \end{eqnarray}
\item Bob sums the 2th and nth qubit to the first component. If the transmitted message is 0, the quantum state becomes: $|0_{1},0_{2},\ldots, 0_{n}\rangle+|1_{1},k_{2},\ldots, k_{n}\rangle$, and then he uses the first bit in the second component $|k_{1}\rangle$ as controlled bit to do CNOT operation to other bits in this state: $|0_{1},0_{2},\ldots, 0_{n}\rangle+|1_{1},0_{2},\ldots, 0_{n}\rangle$. If the transmitted message is 1, the quantum state after operation is: $|0_{1},1_{2},\ldots, 1_{n}\rangle-|1_{1},\overline{k}_{2},\ldots, \overline{k}_{n}\rangle$, and takes the first bit in $\overline{k}_{1}$ as controlled bit to do CNOT operation to other bits to get the state: $|0_{1},0_{2},\ldots, 0_{n}\rangle-|1_{1},0_{2},\ldots, 0_{n}\rangle$. Finally, Bob measures the quantum state with basis $\{|+\rangle, |-\rangle\}$ to get the message.

     \end{enumerate}
\begin{Remark}
 \rm{ The Boolean function $F$ can be generated efficiently. For the m-input, n-output $F=(F^{(1)}(s),\ldots,F^{(n)}(s))$, each output $F^{(i)}(s)$ has $p(n)$ terms. The minor term expression of every $F^{(i)}(s)$ is:}
\begin{eqnarray}
F^{(j)}(s)&=&(s^{a_{j11}}_{1}\cdot \ldots \cdot s^{a_{j1n}}_{n})\oplus\ldots\oplus(s^{a_{jp(n)1}}_{1}\cdot \ldots \cdot s^{a_{jp(n)n}}_{n}).
\end{eqnarray}
 \rm{  Each term $s^{a_{j\alpha1}}_{1}\cdot \ldots \cdot s^{a_{j\alpha n}}_{n}$ can be determined by n times of coin tossing. If we toss the coin for $np(n)$ times, $p(n)$ components are determined. Therefore, the Boolean function $F$ will be efficiently generated by $n^{2}p(n)$ times of coin tossing.}
\end{Remark}
\section{Security analysis}
The security of the bit-oriental quantum public-key encryption scheme proposed above is analyzed from two aspects: (1) the security of private-key; (2) the security of encryption.
\subsection{security of private key}
 The quantum part of public-key is $\rho^{0}_{F(s)}$. Since the adversary has no idea of private-key $F=(F_{1}, F_{2}, F_{3})$, the public-key state for him is $\rho^{0}=\sum\limits_{F}p_{F}\rho^{0}_{F(s)}$. The specific expression of $\rho^{0}_{F(s)}$ is:
 \begin{eqnarray}
 \rho^{0}_{F(S)}&=&Y^{\otimes k_{3}}H^{\otimes k_{2}}\Big[\frac{1}{2}(|0\rangle+|k_{1}\rangle)(\langle 0|+\langle k_{1}|)\Big]H^{\otimes k_{2}}Y^{\otimes k_{3}}.
\end{eqnarray}

Thus, we have:
\begin{eqnarray}
\rho^{0}&=&\sum\limits_{F}p_{F}\rho^{0}_{F(s)}\nonumber\\
 &=&\frac{1}{2^{2n}}\sum\limits_{k_{2}, k_{3}}Y^{\otimes k_{3}}H^{\otimes k_{2}}\Big[\frac{1}{2^{n}}\sum\limits_{k_{1}}(|0\rangle+|k_{1}\rangle)(\langle 0|+\langle k_{1}|)\Big]H^{\otimes k_{2}}Y^{\otimes k_{3}}\nonumber\\
 &=&\frac{1}{2^{2n}}\sum\limits_{k_{2}, k_{3}}Y^{\otimes k_{3}}H^{\otimes k_{2}}\rho(Y^{\otimes k_{3}}H^{\otimes k_{2}})^{\dagger},
\end{eqnarray}
where
\begin{eqnarray}
\rho&=&\frac{1}{2^{n}}\sum\limits_{k_{1}}(|0\rangle+|k_{1}\rangle)(\langle 0|+\langle k_{1}|).
\end{eqnarray}

Since $k_{1}\in\{0, 1\}^{n}$, the possibilities of the value of $k_{1}$ is $2^{n}$. The Boolean function $F_{1}(s_{1})$ has n-output, and each output has $p(n)$ terms, so $n^{2}p(n)$ times coin tossing will determine $F_{1}(s_{1})$. $n^{2}p(n)$ times coin tossing has $2^{n^{2}p(n)}$ possibilities. The possibilities of $F_{1}(s_{1})$ is $2^{n^{n}p(n)-n}$ times of that of $k_{1}$. So the output of private-key $F_{1}(s_{1})$ iterates over all the possible value of $k_{1}$. $(k_{2}, F_{2})$ and $(k_{3}, F_{3})$ is the same case as $(k_{1}, F_{1})$.

Before demonstrate the quantum state of $\rho^{0}$ is an ultimately mixed state, we firstly prove:
$\{p_k=\frac{1}{2^{2n}},U_k=Y^{\alpha}H^{\beta}, k=(\alpha,\beta),\alpha,\beta\in\{0,1\}^n\}$ is a quantum perfect encryption.
\begin{proof}
$\{Y^{\alpha}H^{\beta}, \alpha,\beta\in\{0,1\}^n\}$ is complete orthogonal basis, therefore any n-qubit state can be expressed as a linear combination of $2^{n}$ unitary matrices:
\begin{eqnarray}
\rho&=&\sum\limits_{\alpha, \beta}a_{\alpha, \beta}Y^{\alpha}H^{\beta},
\end{eqnarray}
where
\begin{eqnarray}
a_{\alpha, \beta}&=&\frac{1}{2^{n}}tr(\rho H^{\beta}Y^{\alpha}),
\end{eqnarray}
\begin{eqnarray}
&&\frac{1}{2^{n}}tr( \sum\limits_{\alpha, \beta}a_{\alpha, \beta}Y^{\alpha}H^{\beta}H^{\delta}Y^{\gamma})\nonumber\\ &=&\frac{1}{2^{n}}tr(a_{\alpha, \beta}+\sum\limits_{\gamma\neq\alpha~ or~ \delta\neq\beta}(-1)^{(\delta\oplus\beta)\cdot\gamma}Y^{\alpha+\gamma}H^{\beta+\delta}).
\end{eqnarray}
In the above equation is either
$\alpha\neq\gamma$ or $\beta\neq\delta$, moreover, $tr(Y)=0$, $tr(H)=0$, so $tr\big(\sum\limits_{\gamma\neq\alpha~ or~ \delta\neq\beta}(-1)^{\alpha^{2}+\gamma^{2}+(\delta\oplus\beta)\cdot\gamma}Y^{\alpha+\gamma}H^{\beta+\delta}\big)=0$.
\begin{eqnarray}
\frac{1}{2^{n}}tr( \sum\limits_{\alpha, \beta}a_{\alpha, \beta}Y^{\alpha}H^{\beta}H^{\delta}Y^{\gamma})&=&\frac{1}{2^{n}}a_{\alpha, \beta}
\end{eqnarray}
\begin{eqnarray}
\sum\limits_{k}p_{k}U_{k}\rho U^{\dagger}_{k}&=&\frac{1}{2^{2n}}\sum\limits_{\gamma, \delta}Y^{\gamma}H^{\delta}\rho H^{\delta}Y^{\gamma}\nonumber\\
&=&\frac{1}{2^{2n}}\sum\limits_{\alpha, \beta}a_{\alpha, \beta}\sum\limits_{\gamma, \delta}Y^{\gamma}H^{\delta}Y^{\alpha}H^{\beta}H^{\delta}Y^{\gamma}\nonumber\\
&=&\frac{1}{2^{2n}}\sum\limits_{\alpha, \beta}a_{\alpha, \beta}\sum\limits_{\gamma, \delta}(-1)^{\alpha\cdot\delta}Y^{\alpha}Y^{\gamma}H^{\delta}(-1)^{\beta\cdot\gamma}H^{\delta}Y^{\gamma}H^{\beta}\nonumber\\
&=&\frac{1}{2^{2n}}\sum\limits_{\alpha, \beta}a_{\alpha, \beta}\sum\limits_{\gamma, \delta}(-1)^{\alpha\cdot\delta}(-1)^{\beta\cdot\gamma}Y^{\alpha}H^{\beta},
\end{eqnarray}
because of $\frac{1}{2^{n}}\sum\limits_{\gamma\in\{0,1\}^{n}}(-1)^{\beta\cdot\gamma}=\delta_{\beta, 0}$, $\sum\limits_{k}p_{k}U_{k}\rho U^{\dagger}_{k}$ is as bellow:
\begin{eqnarray}
\sum\limits_{t}p_{k}U_{k}\rho U^{\dagger}_{k}&=&\sum\limits_{\alpha, \beta}a_{\alpha, \beta}\delta_{\alpha, 0}\delta_{\beta, 0}Y^{\alpha}H^{\beta}\nonumber\\
&=&a_{00}I\nonumber\\
&=&\frac{tr(\rho)}{2^{n}}I\nonumber\\
&=&\frac{I}{2^{n}}
\end{eqnarray}
\end{proof}
The proof shows that whatever the concrete state of $\rho$ is , the ciphertext getting by quantum perfect encryption is an ultimately mixed state.
More specifically, the ciphertext has nothing to do with the encrypted key with the condition that  the encyrpted key is unknown to the owner of the ciphertext.

Then we have that $\rho^{0}=\sum\limits_{F}p_{F}\rho^{0}_{F(s)}$ is the ciphertext encypted by quantum perfect encryption:

\begin{eqnarray}
\rho^{0}=\sum\limits_{F}p_{F}\rho^{0}_{F(s)}=\frac{I}{2^{n}}
\end{eqnarray}

The quantum part of public-key is ultimately mixed state for the adversary, so he cannot acquire any information about private key $F=(F_{1}, F_{2}, F_{3})$ by measuring such state. Moreover, according to the proposed scheme, only one copy of each quantum public-key $(s, \rho^{0}_{F(s)})$ is allowed to be generated from the pair of $(s, F)$. Thus, any two public-keys for user is different and one private key corresponds to an exponential number of public-keys. Extracting the value of $k_{1}$, $k_{2}$ and $k_{3}$ by measuring is the same as attacking one-time-pad in classical cryptography. Therefore, extracting the value of $k_{1}$, $k_{2}$ and $k_{3}$ is information-theoretically impossible. Furthermore, extracting the relationship between $s_{1}$, $s_{2}$, $s_{3}$ and $k_{1}$, $k_{2}$, $k_{3}$ are also information-theoretically impossible.
\subsection{security of the encyrption}

Now, we analyze the second aspect and prove the proposed scheme has information-theoretic security.
\begin{Proposition}
The trace distance between any two different ciphertext states encrypted by the same private-key is zero.
\end{Proposition}
\begin{proof}

Let $\rho^{0}_{F(s)}$ be the ciphertext state of message "0" and $\rho^{1}_{F(s)}$ be the ciphertext state of message "1". $\rho^{0}_{F(s)}$ and $\rho^{1}_{F(s)}$ are both encyrpted by the same private-key $F=(F_{1}, F_{2}, F_{3})$. $s=(s_{1}, s_{2}, s_{3})$ is also unchanged for encrypting $\rho^{0}_{F(s)}$ and $\rho^{1}_{F(s)}$. For adversary, he doesn't have private-key, the quantum state for him is $\rho^{b}=\sum\limits_{F}p_{F}\rho^{b}_{F(s)}$.

If the transmitted message is "0", the ciphertext state for the adversary is:
\begin{eqnarray}
\frac{1}{2^{3n}}\sum\limits_{k_{2}, k_{3}}Y^{\otimes k_{3}}H^{\otimes k_{2}}\Big[\sum\limits_{k_{1}}(|0\rangle+|k_{1}\rangle)(\langle 0|+\langle k_{1}|)\Big]H^{\otimes k_{2}}Y^{\otimes k_{3}},
\end{eqnarray}
the ciphertext state, as demonstrated in Sec.4.1, is an ultimately mixed state: $\rho^{0}=\frac{I}{2^{n}}$.

If the message is $1$, the ciphertext state for the adversary is:
\begin{eqnarray}
\frac{1}{2^{2n}}\sum\limits_{n, k_{2}, }Y^{\otimes n}H^{\otimes k_{2}}\Bigg\{\frac{1}{2^{2n}}\sum\limits_{k_{1}, k_{3}}\Big[Y^{\otimes k_{3}}(|0\rangle+|k_{1}\rangle)
(\langle 0|+\langle k_{1}|)Y^{\otimes k_{3}}\Big]\Bigg\}H^{\otimes k_{2}}Y^{\otimes n}
\end{eqnarray}

The ciphertext state for the adversary is:
\begin{eqnarray}
\rho^{1}&=&\sum\limits_{F}p_{F}\rho^{1}_{F(s)}\nonumber\\
&=&\frac{1}{2^{2n}}\sum\limits_{n, k_{2}}Y^{\otimes n}H^{\otimes k_{2}}\Bigg\{\frac{1}{2^{2n}}\sum\limits_{k_{1},k_{3}}\Big[Y^{\otimes k_{3}}(|0\rangle+|k_{1}\rangle)(\langle 0|+\langle k_{1}|)Y^{\otimes k_{3}}\Big]\Bigg\}H^{\otimes k_{2}}Y^{\otimes n}\nonumber\\
&=&\frac{1}{2^{2n}}\sum\limits_{n, k_{2}}Y^{\otimes n}H^{\otimes k_{2}}\rho^{'}(Y^{\otimes n}H^{\otimes k_{2}})^{\dagger},
\end{eqnarray}
\normalsize{where}
\begin{eqnarray}
\rho^{'}&=&\frac{1}{2^{2n}}\sum\limits_{k_{1}, k_{3}}\Big[Y^{\otimes k_{3}}(|0\rangle+|k_{1}\rangle)(\langle 0|+\langle k_{1}|)Y^{\otimes k_{3}}\Big].
\end{eqnarray}
\normalsize{The ciphertext state for message "1"  which can be considered as a quantum perfect encryption is also an ultimately mixed state for adversary:}
\begin{eqnarray}
\rho^{1}&=&\frac{I}{2^{n}}.
\end{eqnarray}
\normalsize{When the same private key is used to encrypt message "0" and "1", the trance distance between $\rho^{0}$ and $\rho^{1}$ is:}
 \begin{eqnarray}
 D(\rho^{0}, \rho^{1})&=&0
\end{eqnarray}
The sufficient condition of indistinguishability of two quantum states is that the trace distance betwenn two quantum is negligible. Because of $D(\rho^{0}, \rho^{1})=0$, it's impossible for the adversary to distinguish $\rho^{0}$ and $\rho^{1}$.

\end{proof}
\begin{Remark}
 \rm{ Assume that adversary intercepts some ciphertext states and attacks private-key by using $\{C_{n}\}$ which is denoted as any quantum circuit family that can distinguish between $\rho^{0}$ and $\rho^{1}$.
Quantum circuit family $\{C_{n}\}$ is a set of unitary transform, and the operation executed by $\{C_{n}\}$ is linear. Given series mixed states $\rho^{b}_{F^1(s)}, \ldots, \rho^{b}_{F^n(s)}$ to $\{C_{n}\}$ one by one as inputs, equals to take $\rho^{b}=\sum\limits_{F}p_{F}\rho^{b}_{F(s)}$ as input. The operation of summing $\rho^{b}_{F(s)}$ over the whole set of private-key $F$ and the operation done by $\{C_{n}\}$ are commutative. That is, the following two inequality are equivalent:}
\begin{eqnarray}
&&\Big|\Pr[C_{n}(\rho^{0})=1]-\Pr[C_{n}(\rho^{1})=1]\Big|\nonumber\\
&=&\Bigg|\Pr\bigg[C_{n}\big(\sum\limits_{F}p_{F}\rho^{0}_{F(s)}\big)=1\bigg]-\Pr\bigg[C_{n}\big(\sum\limits_{F}p_{F}\rho^{1}_{F(s)}\big)=1\bigg]\Bigg|\nonumber\\
&<&\frac{1}{p(n)},
\end{eqnarray}
and
\begin{eqnarray}
\Bigg|\sum\limits_{F}p_{F}\Bigg\{\Pr\bigg[C_{n}\big(\rho^{0}_{F(s)}\big)=1\bigg]-\Pr\bigg[C_{n}\big(\rho^{1}_{F(s)}\big)=1\bigg]\Bigg\}\Bigg|&<&\frac{1}{p(n)}.
\end{eqnarray}
\end{Remark}

\begin{Proposition}
The trace distance between any two different ciphertext states of the same plaintext is zero.
\end{Proposition}
\begin{proof}
Let $\sigma^{b}_{F(s)}$ be the ciphertext state of message "b" encrypted by private-key $F=(F_{1}, F_{2}, F_{3})$ and $\sigma^{b}_{F^{'}(s)}$ be the ciphertext state of message "b" encyrpted by private-key $F^{'}=(F^{'}_{1}, F^{'}_{2}, F^{'}_{3})$. But $s=(s_{1}, s_{2}, s_{3})$ is unchanged for encrypting message $b$. For adversary, he doesn't have private-key, the quantum states for him are $\sigma^{b}=\sum\limits_{F}p_{F}\sigma^{b}_{F(s)}$ and $\sigma^{b}_{1}=\sum\limits_{F^{'}}p_{F^{'}}\sigma^{b}_{F^{'}(s)}$.

If the transmitted message "0" is encrypted by private-key $F^{'}=(F^{'}_{1}, F^{'}_{2}, F^{'}_{3})$,  the ciphertext state for adversary is:
\begin{eqnarray}
\frac{1}{2^{3n}}\sum\limits_{k^{'}_{2}, k^{'}_{3}}Y^{\otimes k^{'}_{3}}H^{\otimes k^{'}_{2}}\Big[\sum\limits_{k^{'}_{1}}(|0\rangle+|k^{'}_{1}\rangle)(\langle 0|+\langle k^{'}_{1}|)\Big]H^{\otimes k^{'}_{2}}Y^{\otimes k^{'}_{3}}.
\end{eqnarray}
As demonstrated in Proposition 1, the above chiphertext state is an ultimately mixed state which has nothing to do with $k^{'}_{1}$ and $k^{'}_{2}$.

 If the transimitted message is $1$ encrypted by private-key $F^{'}=(F^{'}_{1}, F^{'}_{2}, F^{'}_{3})$, the ciphertext state for adversary is:
\begin{eqnarray}
\frac{1}{2^{2n}}\sum\limits_{n, k^{'}_{2}}Y^{\otimes n}H^{\otimes k^{'}_{2}}\Bigg\{\frac{1}{2^{2n}}\sum\limits_{k^{'}_{1}, k^{'}_{3}}Y^{\otimes k^{'}_{3}}\Big[(|0\rangle+|k^{'}_{1}\rangle)                                                                                       (\langle 0|+\langle k^{'}_{1}|)Y^{\otimes k^{'}_{3}}\Big]H^{\otimes k^{'}_{2}}\Bigg\}Y^{\otimes n} ,
\end{eqnarray}
the ciphertext state for message "1" is also an ultimately mixed state for adversary.

Through above analysis, we get the equation:
\begin{eqnarray}
D(\sigma^{b}, \sigma^{b}_{1})&=&0
\end{eqnarray}
If the adversary intercepts two different ciphertext states of the same plaintext, he is unable to distinguish them.
\end{proof}

\begin{Remark}
 When the attacker gets $np(n)$ variables, he can get one bit of private-key in \cite{PY12}. Considering the extreme situation, if the attacker has acquired $n-1$ bis of privat-key, he guesses the last bit wrongly with probability $\frac{1}{2}$ and will get the quantum state as bellow:
 \begin{eqnarray}
|00\rangle+|11\rangle,
\end{eqnarray}
and then he does Hadamard transform on the first qubit:
 \begin{eqnarray}
|00\rangle+|11\rangle\xrightarrow{H}|0\rangle(\frac{|0\rangle+|1\rangle}{\sqrt{2}})+|1\rangle(\frac{|0\rangle-|1\rangle}{\sqrt{2}}).
\end{eqnarray}
When measuring the state, the attacker gets the right one with probability $\frac{1}{2}$. so the probability of the attacker to decide the correct last bit is $\frac{1}{4}$. When attacker is lack of $l$ bits, the probability to get the right private-key is $\frac{1}{2^{l}}+\frac{1}{2}(1-\frac{1}{2^{l}})$.
Since $D(\sigma^{b}, \sigma^{b}_{1})=0$, the attacker cannot attack successfully using this method
\end{Remark}
\section{Discussion}
The proposed bit-oriented QPKE uses private key to do Hadamard transform on the quantum state to hide the information about $i$, $k_{1}$ and $k_{2}$. In \cite{YYX13} pointed out that if the quantum state of public-key is $|i\rangle+|i \oplus k\rangle$, the adversary is able to extract all the information about private-key by applying only random Hadamard transform on public-key state. So the scheme constructed in \cite{PY12} took $|i\rangle+|i \oplus k\rangle$ as its public-key is insecure even in the lower bound of private-key.  The proposed scheme in this paper aims at encrypting classical message and the public-key will be used only once to encrypt one bit. One private-key corresponding to exponential public-keys and one private-key can encrypt exponential classical bits.

For encrypting one-bit, proposition 1 and 2 cannot asure that this shceme information-theoretically secure, we also need to consider the trace distance between $\sum\limits_{F}(\rho^{0}_{F})^{\otimes t}$ and $\sum\limits_{F}\rho^{1}_{F}\otimes(\rho^{0}_{F})^{\otimes (t-1)}$. This situation refers to that the attack intercept $t-1$ public-keys as auxiliary information.
For simplicity, we consider the case $\sum\limits_{F}\rho^{0}_{F}\otimes\rho^{0}_{F}$:
 \begin{eqnarray} 
&& \sum\limits_{F}\rho^{0}_{F}\otimes\rho^{0}_{F}\nonumber\\
 &=&\frac{1}{2^{n}}\sum\limits_{k_{2}}H^{k_{2}}\Big[\frac{1}{2^{2n}}\sum\limits_{k_{1},k_{3}}Y^{k_{3}}(|0\rangle+|k_{1}\rangle)(\langle 0|+\langle k_{1}|)Y^{k_{3}}\nonumber\\
&& \otimes Y^{k_{3}}(|0\rangle+|k_{1}\rangle)(\langle 0|+\langle  k_{1}|)Y^{k_{3}}\Big] H^{k_{2}}\nonumber\\
&=&\frac{1}{2^{3n-1}}\sum\limits_{k_{2}}H^{k_{2}}\Big[\sum\limits_{k_{1}, k_{3}}(|k_{3}, k_{3}\rangle\langle k_{3}, k_{3}|+|k_{1}\oplus k_{3}, k_{3}\rangle\langle k_{1}\oplus k_{3}, k_{3}|\nonumber\\
&&+|k_{1}\oplus k_{3}, k_{1}\oplus k_{3}\rangle\langle k_{3}, k_{3}|+|k_{1}\oplus k_{3}, k_{3}\rangle\langle k_{3}, k_{1}\oplus k_{3}|)\Big]H^{k_{2}},
\end{eqnarray}
Such situation is complex, so $D(\sum\limits_{F}(\rho^{0}_{F})^{\otimes t}, \sum\limits_{F}\rho^{1}_{F}\otimes(\rho^{0}_{F})^{\otimes (t-1)})$ may not be effective computed.

The premise of the bit-oriented QPKE is that the quantum key distribution is secure, therefore we concentrate in the study of QPKE. However, how to ensure the quantum key is distributed securely is an important problem which worth the effort to research on.

Classical PKE relies on one-way function. one-way function \cite{MO97} is based on computational complexity hypothesis. It is a function that it is easy to compute $f(x)$ for every $x$ in the domain of $f$, but the reverse calculation is difficult. More specifically, for all $y$ in the range of $f$, it's hard to find $x$ satisfies $y=f(x)$ in expected polynomial time. We wonder whether QPKE is also based on one-way function, and what's the one-way property of QPKE under the assurance of quantum mechanics.

CCA1, CCA2 and semantically secure are all concerned about the security of classical PKE. The quantum counterparts of these also needs to be further developed.
\section{Conclusions}
This paper studies bit-oriented quantum public-key encryption. The features of our scheme is: $(1)$ private-key is Boolean function, $(2)$ public-key is a pair of classical string and quantum state, $(3)$ one private-key corresponds to an exponential number of public-keys, $(4)$ this scheme ensures information-theoretically security. We analyze the proposed QPKE is information-theoretically secure if every private-key used $o(n)$ times and the security with the cases including attack to the privat-key and attack to the encryption.

\section*{Acknowledgement}
This work was supported by the National Natural Science Foundation of China under Grant No.61173157.

\end{document}